\documentclass[conference]{IEEEtran}

\usepackage[cmex10]{amsmath}
\usepackage[dvipdfmx]{graphicx}
\usepackage{array}
\usepackage{mdwmath}
\usepackage{mdwtab}
\usepackage{amssymb}
\usepackage{float}
\usepackage{subfig}

\newtheorem{theorem}{Theorem}

\newtheorem{definition}[theorem]{Definition}

\newtheorem{remark}[theorem]{Remark}
\newtheorem{proposition}[theorem]{Proposition}
\newcommand{\qed}{\hfill$\square$}

\newenvironment{proof}{%
 \noindent{\em Proof.\ }}{%
 \hspace*{\fill}\qed \\
 \vspace{2ex}}
\DeclareMathAlphabet{\bm}{OML}{cmm}{b}{it}

\newcommand{\san}[1]{\mathsf{#1}}

\begin{document}
%
\title{On Sub-optimality of Random Binning for Distributed Hypothesis Testing}



%
\author{\IEEEauthorblockN{Shun Watanabe\IEEEauthorrefmark{1}}
\IEEEauthorblockA{\IEEEauthorrefmark{1}Department of Computer and Information Sciences, 
Tokyo University of Agriculture and Technology, Japan, \\
E-mail:shunwata@cc.tuat.ac.jp}}


\maketitle

\begin{abstract}
We investigate the quantize and binning scheme, known as the Shimokawa-Han-Amari (SHA) scheme,
for the distributed hypothesis testing.
We develop tools to evaluate the critical rate attainable by the SHA scheme.
For a product of binary symmetric double sources,
we present a sequential scheme that improves upon the SHA scheme.
\end{abstract}


%
\IEEEpeerreviewmaketitle

\section{Introduction} \label{section:introduction}

In \cite{berger:79b}, Berger introduced a framework of statistical decision problems under communication constraint.
Inspired by his work, many researchers studied various problems of this kind \cite{AhlCsi86, han:87, ZhaBer88, amari-han:89, amari:89, han-kobayashi:89, ShaPap92};
see \cite{han-amari:98} for a thorough review until 90s. More recently, the problem of communication constrained statistics has regained interest of researchers;
see \cite{Amari:11, ueta-kuzuoka:14, polyanskiy:12, RahWag:12, XiaKim:13, ZhaLai:18, Wat:18, WeiKoc:19, SalWigWan:19, HadLiuPolSha:19, HadSha:19, EscWigZai:20, Bur:20, SreGun:20, SahTya:21}
 and references therein.

In this paper, we are interested in the
most basic setting of hypothesis testing in which the sender and the receiver observe correlated sources $X^n$ and $Y^n$, respectively,
and the observation $(X^n,Y^n)$ is distributed according to the product of the null hypothesis $P_{XY}$ or the alternative
hypothesis $Q_{XY}$; the sender transmit a message at rate $R$, and the receiver decide the hypotheses based on the message and its observation.

In \cite{AhlCsi86}, Ahlswede and Csisz\'ar introduced the so-called quantization scheme, 
and showed that the quantization scheme is optimal for the testing against independence, i.e., 
the alternative hypothesis is $Q_{XY} = P_X\times P_Y$. In \cite{han:87}, Han introduced an improved 
version of the quantization scheme that fully exploit the knowledge of the marginal types. 
In \cite{ShiHanAma94}, Shimokawa, Han, and Amari introduced the quantize and binning scheme, which we shall
call the SHA scheme in the following. For a long time, it has not been known if the SHA is optimal or not.
In \cite{RahWag:12}, Rahman and Wagner showed that the SHA scheme is optimal for the testing against conditional 
independence.\footnote{In fact, for the testing against conditional independence, it suffices to consider a simpler version of quantize and binning scheme 
than the SHA scheme. } More recently, Weinberger and Kochman \cite{WeiKoc:19} (see also \cite{HaiKoch:16}) derived the optimal exponential
trade-off between the type I and type II error probabilities of the quantize and binning scheme. 
Since the decision is conducted based on the optimal likelihood decision rule in \cite{WeiKoc:19}, their scheme may potentially provide 
a better performance than the SHA scheme; however, since the expression involves complicated optimization over multiple parameters, and a strict improvement has 
not been clarified so far.

A main objective of this paper is to investigate if the SHA scheme is optimal or not. In fact, it turns out that the answer is negative,
and there exists a scheme that improves upon the SHA scheme. 
Since the general SHA scheme is involved, we focus on the critical rate, i.e., the minimum rate that is required to attain the same type II exponent
as the Stein exponent of the centralized testing. Then, under a certain non-degeneracy condition,
we will prove that the auxiliary random variable in the SHA scheme must coincide with the source $X$ itself, i.e., we must use the binning without quantization. 
By leveraging this simplification, for the binary symmetric double sources (BSDS), we derive a closed form expression of the critical rate attainable by the SHA scheme.
Next, we will consider a product of the BSDS such that each component of the null hypothesis and the alternative hypothesis is ``reversely aligned."
Perhaps surprisingly, it turns out that the SHA is sub-optimal for such hypotheses. Our improved scheme is simple, and it uses the SHA scheme 
for each component in a sequential manner; however, it should be emphasized that we need to conduct binning of each component separately,
and our scheme is not a naive random binning. 

\section{Problem Formulation}

We consider the statistical problem of testing the null hypothesis $\mathtt{H}_0: P_{XY}$ on finite alphabets ${\cal X}\times {\cal Y}$ versus the alternative 
hypothesis $\mathtt{H}_1: Q_{XY}$ on the same alphabet. The i.i.d. random variables $(X^n,Y^n)$ distributed according to either $P_{XY}^n$ or $Q_{XY}^n$
are observed by the sender and the receiver; the sender encodes the observation $X^n$ to a message by an encoding function
\begin{align*}
\varphi_n: {\cal X}^n \to {\cal M}_n,
\end{align*} 
and the receiver decides whether to accept the null hypothesis or not by a decoding function
\begin{align*}
\psi_n:{\cal M}_n \times {\cal Y}^n \to \{\mathtt{H}_0,\mathtt{H}_1 \}.
\end{align*}
When the block length $n$ is obvious from the context, we omit the subscript $n$.
For a given testing scheme $T_n = (\varphi,\psi)$, the type I error probability is defined by
\begin{align*}
\alpha[T_n] := P\bigg( \psi(\varphi(X^n),Y^n) = \mathtt{H}_1 \bigg)
\end{align*}
and the type II error probability is defined by 
\begin{align*}
\beta[T_n] := Q\bigg( \psi(\varphi(X^n),Y^n) = \mathtt{H}_0 \bigg).
\end{align*}
In the following, $P(\cdot)$ (or $Q(\cdot)$) means that $(X^n,Y^n)$ is distributed according to $P_{XY}^n$ (or $Q_{XY}^n$).

In the problem of distributed hypothesis testing, we are interested in the trade-off among the communication rate
and the type I and type II error probabilities. Particularly, we focus on
the optimal trade-off between the communication rate and the exponent of the type II error probability
when the type I error probability is vanishing, which is sometimes called ``Stein regime".
\begin{definition}
A pair $(R,E)$ of the communication rate and the exponent is defined to be achievable if there exists a sequence $\{T_n\}_{n=1}^\infty$
of testing schemes such that
\begin{align*}
\lim_{n\to\infty} \alpha[T_n] = 0
\end{align*}
and
\begin{align*}
\limsup_{n\to\infty} \frac{1}{n} \log |{\cal M}_n| &\le R, \\
\liminf_{n\to\infty} - \frac{1}{n} \log \beta_n[T_n] &\ge E.
\end{align*}
The achievable region ${\cal R}$ is the set of all achievable pair $(R,E)$. 
\end{definition}

For a given communication rate $R$, the maximum exponent is denoted by 
\begin{align*}
E(R) := \max\{ E : (R,E) \in {\cal R} \}.
\end{align*}
When there is no communication constraint, i.e., $R \ge \log |{\cal X}|$, the sender can send
the full observation without any encoding, and the receiver can use a scheme for the standard
(centralized) hypothesis testing. In such a case, the Stein exponent $D(P_{XY}\|Q_{XY})$ is attainable. 
We define the critical rate as the minimum communication rate required to attain the Stein exponent:
\begin{align*}
R_{\mathtt{cr}} := \inf\{ R : E(R) = D(P_{XY}\|Q_{XY}) \}.
\end{align*}

\section{On Evaluation of SHA Scheme}

When the communication rate $R$ is not sufficient to send $X$, the standard approach is to quantize $X$ using
a test channel $P_{U|X}$ for some finite auxiliary alphabet ${\cal U}$. For $P_{UXY} = P_{U|X}P_{XY}$
and $Q_{UXY} = Q_{U|X} Q_{XY}$ with $Q_{U|X} = P_{U|X}$, let
\begin{align*}
\lefteqn{ E(P_{UXY} \| Q_{UXY}) } \\ 
&:= \min\big\{ D(P_{\tilde{U}\tilde{X}\tilde{Y}} \| Q_{UXY}) 
: P_{\tilde{U}\tilde{X}} = P_{UX}, P_{\tilde{U}\tilde{Y}}=P_{UY} \big\}.
\end{align*}
In \cite{han:87}, a testing scheme based on quantization was proposed, and the following achievability bound was derived: for any test channel $P_{U|X}$, 
\begin{align*}
E(R) \ge E(P_{UXY} \| Q_{UXY}).
\end{align*}

In \cite{ShiHanAma94}, a testing scheme based on quantization and binning was proposed,
and the following achievability bound was derived: for any test channel $P_{U|X}$ such that $R \ge I(U \wedge X|Y)$, 
\begin{align*}
E(R) &\ge \min\bigg[ \min_{P_{\tilde{U}\tilde{X}\tilde{Y}} \in{\cal P}_{\mathtt{b}}(P_{U|X})} D(P_{\tilde{U}\tilde{X}\tilde{Y}} \| Q_{UXY}) \\
&~~~ ~~~+ | R - I(U \wedge X|Y)|^+, E(P_{UXY} \| Q_{UXY}) \bigg], 
\end{align*}
where $|t|^+ = \max[t,0]$, and 
\begin{align*}
{\cal P}_{\mathtt{b}}(P_{U|X}) &:= \big\{ P_{\tilde{U}\tilde{X}\tilde{Y}} : \\
&~~~ P_{\tilde{U}\tilde{X}} = P_{UX}, P_{\tilde{Y}} = P_Y, H(\tilde{U}|\tilde{Y}) \ge H(U|Y) \big\}.
\end{align*}
Particularly, by taking $P_{U|X}$ to be noiseless test channel, i.e., $U=X$, we have
\begin{align} \label{eq:b-X}
E(R) &\ge E_{\mathtt{b}}(R)  \\
&:= \min\bigg[ \min_{P_{\tilde{X}\tilde{Y}} \in {\cal P}_{\mathtt{b}}} D(P_{\tilde{X}\tilde{Y}} \| Q_{XY}) + |R - H(X|Y)|^+, \nonumber \\
&~~~~~~ D(P_{XY}\|Q_{XY}) \bigg], 
\end{align}
where 
\begin{align*}
{\cal P}_{\mathtt{b}} := \big\{ P_{\tilde{X}\tilde{Y}} : P_{\tilde{X}}=P_X, P_{\tilde{Y}}=P_Y, H(\tilde{X}|\tilde{Y}) \ge H(X|Y) \big\}.
\end{align*}

In the following, we focus on the critical rate. Under some mild conditions, in order to attain the Stein exponent,
we must take $U=X$ in the quantization-bining bound, i.e., no quantization. In the rest of this section, we prove this claim. 

Suppose that $P_{XY}$ and $Q_{XY}$ have full support. Let $\hat{\Lambda}$ be a function on ${\cal X} \times {\cal Y}$ defined by\footnote{We suppose that
alphabets are ${\cal X}=\{0,\ldots,m_\san{x}\}$ and ${\cal Y}=\{0,\ldots,m_\san{y}\}$; taking $y=0$ as a special symbol in \eqref{eq:LLR}
is just notational convenience.}
\begin{align} \label{eq:LLR}
\hat{\Lambda}(x,y) := \log \frac{P_{XY}(x,y)}{Q_{XY}(x,y)} - \log \frac{P_{XY}(x,0)}{Q_{XY}(x,0)}.
\end{align}

\begin{proposition} \label{prop:no-quantization}
Suppose that, for every $x\neq x^\prime$, rows $\hat{\Lambda}(x,\cdot)$ and $\hat{\Lambda}(x^\prime,\cdot)$ are distinct.
Then, $E(P_{UXY} \| Q_{UXY}) = D(P_{XY} \| Q_{XY})$ only if there does not exist $u\in {\cal U}$ and $x\neq x^\prime$ such that
$P_{U|X}(u|x) P_{U|X}(u|x^\prime) > 0$.
\end{proposition}
\begin{proof}
Since we assumed that $P_{XY}$ and $Q_{XY}$ have full support and since $P_{U|X}=Q_{U|X}$, 
$P_{UXY}$ and $Q_{UXY}$ have the same support, which we denote by ${\cal A}$.
Let ${\cal A}^\prime = \mathtt{supp}(P_{UX}) = \mathtt{supp}(Q_{UX})$. Note that $P_{UY}$ and $Q_{UY}$ have full support.
Let ${\cal L} \subseteq {\cal P}({\cal A})$ be the linear family of distributions $P_{\tilde{U}\tilde{X}\tilde{Y}}$ satisfying 
\begin{align*}
\sum_{(u,x,y) \in {\cal A}} P_{\tilde{U}\tilde{X}\tilde{Y}}(u,x,y) \delta_{\bar{u}\bar{x}}(u,x) = P_{UX}(\bar{u},\bar{x})
\end{align*}
for every $(\bar{u},\bar{x}) \in {\cal A}^\prime$ and 
\begin{align*}
\sum_{(u,x,y) \in {\cal A}} P_{\tilde{U}\tilde{X}\tilde{Y}}(u,x,y) \delta_{\bar{u}\bar{y}}(u,y) = P_{UY}(\bar{u},\bar{y}),
\end{align*}
for every $(\bar{u},\bar{y}) \in {\cal U} \times {\cal Y}$.
Then, we can write $E(P_{UXY}\|Q_{UXY})$ as 
\begin{align} \label{eq:minimization-problem}
E(P_{UXY}\|Q_{UXY}) = \min_{P_{\tilde{U}\tilde{X}\tilde{Y}} \in {\cal L}} D(P_{\tilde{U}\tilde{X}\tilde{Y}}\|Q_{UXY}).
\end{align}
If 
\begin{align*}
E(P_{UXY} \| Q_{UXY}) = D(P_{XY} \| Q_{XY}) = D(P_{UXY}\|Q_{UXY}), 
\end{align*}
then $P_{UXY}$ must be the (unique) minimizer in \eqref{eq:minimization-problem}. 
Then, the I-projection theorem (cf.~\cite[Theorem 3.2]{csiszar-shields:04book}) 
implies that the minimizer $P_{UXY}$ lies in the exponential family generated by 
$Q_{UXY}$ and the constraints of ${\cal L}$, i.e., it has of the form
\begin{align*}
\lefteqn{ P_{UXY}(u,x,y) }\\
&\propto Q_{UXY}(u,x,y) \exp\bigg[ \sum_{\bar{u},\bar{x}} \theta_{\bar{u}\bar{x}} \delta_{\bar{u}\bar{x}}(u,x) 
+ \sum_{\bar{u},\bar{y}} \theta_{\bar{u}\bar{y}} \delta_{\bar{u}\bar{y}}(u,y) \bigg]. 
\end{align*}
for some $\theta_{\bar{u}\bar{x}},\theta_{\bar{u}\bar{y}} \in \mathbb{R}$. Equivalently, this condition can be written as 
\begin{align*}
\log \frac{P_{UXY}(u,x,y)}{Q_{UXY}(u,x,y)} = g_1(u,x) + g_2(u,y)
\end{align*}
for some functions $g_1$ on ${\cal A}^\prime$ and $g_2$ on ${\cal U}\times {\cal Y}$. If there exist $u\in {\cal U}$ and $x\neq x^\prime$ such that
$P_{U|X}(u|x) P_{U|X}(u|x^\prime) > 0$, then, for every $y\in {\cal Y}$, we have
\begin{align*}
\lefteqn{ \hat{\Lambda}(x,y) - \hat{\Lambda}(x^\prime,y) } \\
&= \bigg[ \log\frac{P_{XY}(x,y)}{Q_{XY}(x,y)} - \log \frac{P_{XY}(x,0)}{Q_{XY}(x,0)} \bigg] \\
&~~~ - \bigg[ \log \frac{P_{XY}(x^\prime,y)}{Q_{XY}(x^\prime,y)} - \log \frac{P_{XY}(x^\prime,0)}{Q_{XY}(x^\prime,0)} \bigg] \\
&= \bigg[ \log\frac{P_{UXY}(u,x,y)}{Q_{UXY}(u,x,y)} - \log \frac{P_{UXY}(u,x,0)}{Q_{UXY}(u,x,0)} \bigg] \\
&~~~ - \bigg[ \log \frac{P_{UXY}(u,x^\prime,y)}{Q_{UXY}(u,x^\prime,y)} - \log \frac{P_{UXY}(u,x^\prime,0)}{Q_{UXY}(u,x^\prime,0)} \bigg] \\
&= \bigg[ g_1(u,x) + g_2(u,y) - g_1(u,x) - g_2(u,0) \bigg] \\
&~~~ - \bigg[ g_1(u,x^\prime) + g_2(u,y) - g_1(u,x^\prime) - g_2(u,0) \bigg] \\
&=0,
\end{align*}
which contradict the assumption that $\hat{\Lambda}(x,\cdot)$ and $\hat{\Lambda}(x^\prime,\cdot)$ are distinct. 
\end{proof}

Since the quantization-binning bound is at most $E(P_{UXY}\|Q_{UXY})$, when the assumption of Proposition \ref{prop:no-quantization} is satisfied, 
$U=X$ is the only choice of test channel to attain the Stein exponent.\footnote{Even though we can take 
a redundant test channel $P_{U|X}$ such that $\mathtt{supp}(P_{U|X}(\cdot|x)) \cap \mathtt{supp}(P_{U|X}(\cdot|x^\prime)) = \emptyset$
for any $x \neq x^\prime$, such a test channel does not improve the attainable exponent.} 

In Proposition \ref{prop:no-quantization},
the assumption that every raws of $\hat{\Lambda}(x,y)$ are distinct cannot be replaced by the assumption
that every raws of the log-likelihood function
\begin{align*}
\Lambda(x,y) = \log \frac{P_{XY}(x,y)}{Q_{XY}(x,y)}
\end{align*}
are distinct; for instance, when $P_{XY}$ and $Q_{XY}$ are given by
\begin{align*}
P_{XY} &= \frac{1}{24a} \left[
\begin{array}{ccc}
a & 2 a & 3a \\ 2a & 3a & 3a \\ a & 3a & 6a
\end{array}
\right], \\
Q_{XY} &= \frac{1}{14a}
\left[
\begin{array}{ccc}
2a & a & a \\ a & a & a \\ a & 2a & 4a
\end{array}
\right]
\end{align*}
for some $a >0$, then every raws $\Lambda(x,\cdot)$ are distinct but
$\hat{\Lambda}(1,\cdot) = \hat{\Lambda}(2,\cdot)$. In fact, as is shown in the following Proposition \ref{prop:merged},
the symbols $x=1$ and $x=2$ in this example can be merged while maintaining 
$E( P_{UXY} \| Q_{UXY}) = D(P_{XY} \| Q_{XY})$. Even though Proposition \ref{prop:merged} is not used in the rest
of the paper, it may be of independent interest. 

\begin{proposition} \label{prop:merged}
Let $\kappa:{\cal X} \to {\cal U}$ be a function such that $\kappa(x)=\kappa(x^\prime)$
if and only if $\hat{\Lambda}(x,\cdot)=\hat{\Lambda}(x^\prime,\cdot)$, and let $U = \kappa(X)$.
Then, for any $P_{\tilde{U}\tilde{X}\tilde{Y}}$ satisfying $P_{\tilde{U}\tilde{X}} = P_{UX}$ and $P_{\tilde{U}\tilde{Y}}=P_{UY}$, we have
\begin{align}
\lefteqn{ D(P_{\tilde{U}\tilde{X}\tilde{Y}} \| Q_{UXY}) } \nonumber \\
&= D(P_{\tilde{U}\tilde{X}\tilde{Y}} \| P_{UXY}) + D(P_{UXY} \| Q_{UXY}). \label{eq:Phytagorean}
\end{align}
Particularly, we have
\begin{align} \label{eq:no-loss}
E( P_{UXY} \| Q_{UXY}) = D(P_{XY} \| Q_{XY}).
\end{align}
\end{proposition}
\begin{proof}
Since \eqref{eq:no-loss} follows from \eqref{eq:Phytagorean}, it suffices to prove \eqref{eq:Phytagorean}.
For each $u$, fix $x_u\in {\cal X}$ such that $\kappa(x_u)=u$. Then, we have
\begin{align*}
\lefteqn{ D(P_{\tilde{U}\tilde{X}\tilde{Y}} \| Q_{UXY} ) - D(P_{\tilde{U}\tilde{X}\tilde{Y}} \| P_{UXY}) - D(P_{UXY} \| Q_{UXY}) } \\
&= \sum_{u,x,y} \big( P_{\tilde{U}\tilde{X}\tilde{Y}}(u,x,y) - P_{UXY}(u,x,y) \big) \log \frac{P_{UXY}(u,x,y)}{Q_{UXY}(u,x,y)} \\
&= \sum_{u,x,y} \big( P_{\tilde{U}\tilde{X}\tilde{Y}}(u,x,y) - P_{UXY}(u,x,y) \big) \log \frac{P_{XY}(x,y)}{Q_{XY}(x,y)} \\
&= \sum_{u,x,y} \big( P_{\tilde{U}\tilde{X}\tilde{Y}}(u,x,y) - P_{UXY}(u,x,y) \big) \hat{\Lambda}(x,y) \\
&= \sum_{u,x,y} \big( P_{\tilde{U}\tilde{X}\tilde{Y}}(u,x,y) - P_{UXY}(u,x,y) \big) \hat{\Lambda}(x_u,y) \\
&= \sum_{u,y} \big( P_{\tilde{U}\tilde{Y}}(u,y) - P_{UY}(u,y) \big) \hat{\Lambda}(x_u,y) \\
&=0,
\end{align*}
where the third equality follows from $P_{\tilde{X}}=P_X$ and the last equality follows from $P_{\tilde{U}\tilde{Y}}=P_{UY}$.
\end{proof}

\begin{remark}
Introducing the function $\Lambda(x,y)$ as in \eqref{eq:LLR} is motivated by the problem of distributed function computing.
Note that, in order to conduct the likelihood ratio test, the receiver need not compute the
log-likelihood ratios $(\Lambda(x_i,y_i) : 1 \le i \le n)$ for a given observation $(\bm{x},\bm{y}) \in {\cal X}^n \times {\cal Y}^n$; 
instead, it suffices to compute the summation
\begin{align*}
\Lambda_n(\bm{x},\bm{y}) = \sum_{i=1}^n \Lambda(x_i,y_i).
\end{align*}
In the terminology of distributed computing \cite[Definition 4]{KuzWat16}, 
we can verify that the function $\Lambda_n$ is 
$\overline{{\cal X}}_{\hat{\Lambda}}$-informative. 
\end{remark}

\section{Binary Symmetric Double Source (BSDS)}

Let ${\cal X}={\cal Y} = \{0,1\}$, and let $P_{XY}$ and $Q_{XY}$ be given by
\begin{align*}
P_{XY} = \left[
\begin{array}{cc}
\frac{1-p}{2} & \frac{p}{2} \\
\frac{p}{2} & \frac{1-p}{2}
\end{array}
\right],~~~
Q_{XY} = \left[
\begin{array}{cc}
\frac{1-q}{2} & \frac{q}{2} \\
\frac{q}{2} & \frac{1-q}{2}
\end{array}
\right]
\end{align*}
for $0\le p,q \le 1$ and $p \neq q$.
Since $\hat{\Lambda}(0,0)=\hat{\Lambda}(1,0)=0$ and 
\begin{align*}
\hat{\Lambda}(0,1) = \log \frac{(1-q)p}{(1-p)q} \neq  \log \frac{(1-p)q}{(1-q)p} = \hat{\Lambda}(1,1),
\end{align*}
the assumption of Proposition \ref{prop:no-quantization} is satisfied. Thus, we need to set $U=X$ in the quantization-binning bound to attain the Stein exponent.
In the minimization of \eqref{eq:b-X}, since $P_{\tilde{X}}=P_X$ and $P_{\tilde{Y}}=P_Y$ are the uniform distribution, the only free parameter is the crossover probability
$P_{\tilde{Y}|\tilde{X}}(1|0)=P_{\tilde{Y}|\tilde{X}}(0|1)=\tilde{p}$. Then, we can easily find that the binning bound for BSDS is
\begin{align}  
\lefteqn{ E_{\mathtt{b}}(R) } \nonumber \\ 
&=  \left\{
\begin{array}{ll}
\min[ |R - h(p)|^+, D(p\|q) ] & \mbox{ if } h(p) \le h(q) \\
D(p\| q) & \mbox{ if } h(p) > h(q)
\end{array}
\right. \label{eq:b-Exponent-BSDS}
\end{align} 
for $R \ge h(p)$. Furthermore, this bound implies that the critical rate attainable by the SHA scheme is 
\begin{align} \label{eq:b-critical-BSDS}
R_{\mathtt{cr}} \le  \left\{
\begin{array}{ll}
h(p) + D(p\|q) & \mbox{ if } h(p) \le h(q) \\
h(p) & \mbox{ if } h(p) > h(q)
\end{array}
\right..
\end{align}

\section{Product of BSDS} \label{section:product-BSDS}

Next, let us consider the case where $P_{XY} = P_{X_1 Y_1} \times P_{X_2 Y_2}$
and $Q_{XY} = Q_{X_1 Y_1} \times Q_{X_2 Y_2}$, and $P_{X_i Y_i}$ and $Q_{X_i Y_i}$ for $i=1,2$
are DSBSs with parameters $0 \le p_1,p_2,q_1,q_2 \le 1$, respectively. Particularly, we consider the case where $p_1 = q_2$,
$p_2 = q_1$, where $p_1 \neq q_1$. Note that, for the product source, the function defined by \eqref{eq:LLR} can be decomposed as
\begin{align*}
\hat{\Lambda}(x_1x_2,y_1y_2) = \hat{\Lambda}_1(x_1,y_1) + \hat{\Lambda}_2(x_2,y_2),
\end{align*}
where
\begin{align*}
\hat{\Lambda}_i(x_i,y_i) = \log \frac{P_{X_i Y_i}(x_i,y_i)}{Q_{X_iY_i}(x_i,y_i)} - \log \frac{P_{X_i Y_i}(x_i,0)}{Q_{X_iY_i}(x_i,0)}.
\end{align*}
Furthermore, we have $\hat{\Lambda}_1(1,1) = - \hat{\Lambda}_1(0,1)$, 
$\hat{\Lambda}_2(0,1)=-\hat{\Lambda}_1(0,1)$, and $\hat{\Lambda}_2(1,1)=\hat{\Lambda}_1(0,1)$. Thus, by denoting $a = \hat{\Lambda}_1(0,1)$, we have
\begin{align*}
\hat{\Lambda}(x_1x_2,y_1y_2) = 
\left[
\begin{array}{cccc}
0 & - a & a & 0 \\
0 & a & a & 2 a \\
0 & - a & - a & - 2 a \\
0 & a & - a & 0
\end{array}
\right],
\end{align*}
and the assumption of Proposition \ref{prop:no-quantization} is satisfied. Thus, we need to set $U=X$ in the quantization-binning bound to attain the Stein exponent.

Since $p_1=q_2$ and $p_2=q_1$, the conditional entropies $H_P(X|Y)$ and $H_Q(X|Y)$ with respect to $P_{XY}$ and $Q_{XY}$ satisfy
$H_P(X|Y)=H_Q(X|Y)$. Thus, we find that the inner minimization of the binning bound is
attained by $P_{\tilde{X}\tilde{Y}} = Q_{XY}$, and 
\begin{align*}
\lefteqn{ E_{\mathtt{b}}(R) } \\
&= \min\big[ |R- h(p_1) - h(p_2)|^+, D(p_1\|q_1) + D(p_2\|q_2) \big].
\end{align*}
Furthermore, this bound implies that the critical rate attainable by the SHA scheme is
\begin{align} \label{eq:critical-product-SHA}
R_{\mathtt{cr}} \le h(p_1) + h(p_2) + D(p_1\|q_1) + D(p_2 \| q_2).
\end{align}

\section{Sequential Scheme for Product of BSDS} \label{section:sequential}

Now, we describe a modified version of SHA scheme for the product of BSDS. 
For $i=1,2$, let $T_{\mathtt{b}}^{[i]} = (\varphi^{[i]},\psi^{[i]})$ be the SHA scheme (without quantization) for $P_{X_iY_i}$ versus $Q_{X_iY_i}$.
For a given rate $R$, we split the rate as $R=R_1 + R_2$, and consider sequential scheme $\tilde{T}_{\mathtt{b}} = (\varphi,\psi)$ 
constructed from $(T_{\mathtt{b}}^{[1]},T_{\mathtt{b}}^{[2]})$ as follows.
Upon observing $\bm{x} = (\bm{x}_1,\bm{x}_2)$,
the sender transmit $m_1 = \varphi^{[1]}(\bm{x}_1)$ and $m_2 = \varphi^{[2]}(\bm{x}_2)$.\footnote{Note that $m_1$ and $m_2$ include
indices of random binning as well as marginal types of $\bm{x}_1$ and $\bm{x}_2$.}
Then, upon receiving $(m_1,m_2)$ and observing $\bm{y}=(\bm{y}_1,\bm{y}_2)$, the receiver decides 
$\psi(m_1,m_2,\bm{y}) = \mathtt{H}_0$ if $\psi^{[1]}(m_1,\bm{y}_1) = \mathtt{H}_0$ and $\psi^{[2]}(m_2,\bm{y}_2) = \mathtt{H}_0$.
The type I and type II error probabilities of this sequential scheme $\tilde{T}_{\mathtt{b}}$ can be evaluated as 
\begin{align*} 
\lefteqn{ \alpha[\tilde{T}_{\mathtt{b}}] } \\ 
&= P\bigg( \psi^{[1]}(\varphi^{[1]}(X_1^n),Y_1^n) = \mathtt{H}_1 \vee \psi^{[2]}(\varphi^{[2]}(X_2^n),Y_2^n) = \mathtt{H}_1 \bigg) \nonumber \\
&\le \sum_{i=1}^2 P\bigg( \psi^{[i]}(\varphi^{[i]}(X_i^n),Y_i^n) = \mathtt{H}_1 \bigg)
\end{align*}
by the union bound, and
\begin{align*} 
\lefteqn{ \beta[\tilde{T}_{\mathtt{b}}] } \\ 
&= Q\bigg(  \psi^{[1]}(\varphi^{[1]}(X_1^n),Y_1^n) = \mathtt{H}_0,~ \psi^{[2]}(\varphi^{[2]}(X_2^n),Y_2^n) = \mathtt{H}_0 \bigg) \nonumber \\
&= \prod_{i=1}^2 Q\bigg( \psi^{[i]}(\varphi^{[i]}(X_i^n),Y_i^n) = \mathtt{H}_0 \bigg).
\end{align*}
Thus, if the type I error probabilities of each scheme is vanishing, then the type I error probability $\alpha[\tilde{T}_{\mathtt{b}}]$
is also vanishing. On the other hand, the type II exponent of the sequential scheme is the summation of
the type II exponent of each scheme.

As in Section \ref{section:product-BSDS}, suppose that $p_1=q_2$ and $p_2=q_1$. Furthermore, let $h(q_1) < h(p_1)$. Then, 
from the above argument,
and \eqref{eq:b-Exponent-BSDS}, we  find that the following exponent is attainable by the sequential scheme: 
\begin{align}
\tilde{E}_{\mathtt{b}}(R) = D(p_1 \| q_1) + \min[ |R_2 - h(p_2)|^+, D(p_2 \| q_2) ]
\end{align}
for $R_1 \ge h(p_1)$ and $R_2 \ge h(p_2)$.
By setting $R_1 = h(p_1)$ and $R_2 = h(p_2) + D(p_2 \| q_2)$, we can derive the following bound on the critical rate:
\begin{align} \label{eq:critical-sequential-BSDS}
R_{\mathtt{cr}} \le h(p_1) + h(p_2) + D(p_2 \| q_2).
\end{align}
Interestingly, the bound on the critical rate in \eqref{eq:critical-sequential-BSDS} improves upon
the bound in \eqref{eq:critical-product-SHA}. 

From an operational point of view, this improvement 
can be explained as follows. In the SHA scheme (without quantization), we use the minimum entropy decoder to compute an
estimate $\hat{\bm{x}}$ of $X^n =\bm{x}$; then, if the joint type $\san{t}_{\hat{\bm{x}} \bm{y}}$ is close to $P_{XY}$, we accept the null
hypothesis. When the empirical conditional entropy $H(\bm{x}|\bm{y})$ is such that $H_P(X|Y) \gnsim H(\bm{x}|\bm{y})$,
even if there exists $\hat{\bm{x}}$ satisfying $H(\bm{x}|\bm{y}) \ge H(\hat{\bm{x}}|\bm{y})$, the type II testing error does not occur
though the receiver may erroneously compute $\hat{\bm{x}}\neq \bm{x}$.\footnote{Note that $H_P(X|Y) \gnsim H(\bm{x}|\bm{y}) \ge H(\hat{\bm{x}}|\bm{y})$
implies that $\san{t}_{\hat{\bm{x}} \bm{y}}$ is not close to $P_{XY}$.} 
Consequently, the type II testing error may occur 
only if $(X^n,Y^n) = (\bm{x},\bm{y})$ satisfies $H(\bm{x}|\bm{y}) \gtrsim H_P(X|Y)$.

For the product of  BSDSs, when the types of modulo sums $\bm{x}_1 + \bm{y}_1$ and $\bm{x}_2 + \bm{y}_2$
are $H(\san{t}_{\bm{x}_1 + \bm{y}_1}) \simeq h(\tau_1)$ and $H(\san{t}_{\bm{x}_2+\bm{y}_2}) \simeq h(\tau_2)$,
the type II testing error may occur only if $h(\tau_1) + h(\tau_2) \gtrsim h(p_1) + h(p_2)$; see Fig. \ref{Fig:range-product1}(a).
In the modified SHA scheme, $\bm{x}_1$ and $\bm{x}_2$ are binned and estimated separately. Because of this, 
the constraint of type II testing error event become more stringent, i.e., $h(\tau_1) \gtrsim h(p_1)$ and $h(\tau_2) \gtrsim h(p_2)$; see Fig.~Fig. \ref{Fig:range-product1}(b). 

\begin{figure}[t]
\centering{
 \subfloat[][]{\includegraphics[width=.4\textwidth]{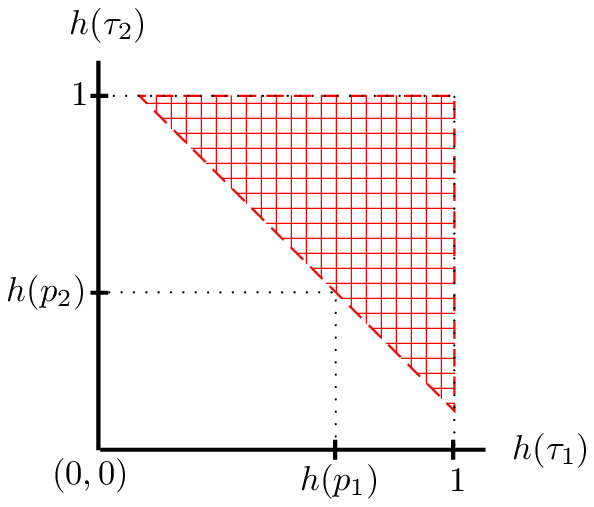} \label{Fig:range-product1}}  \quad
 \subfloat[][]{\includegraphics[width=.4\textwidth]{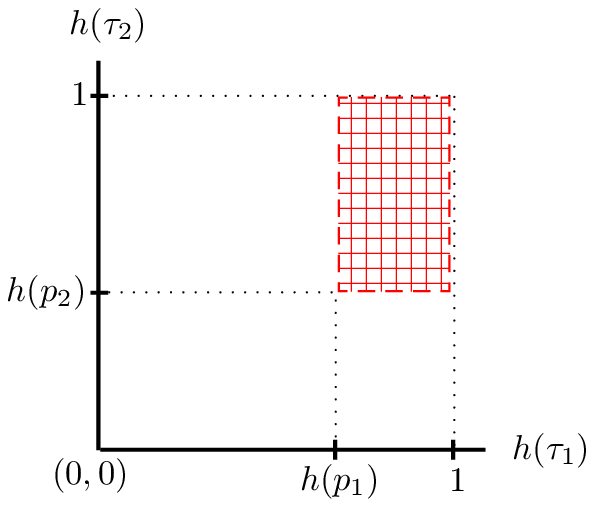} \label{Fig:range-product2}}
\caption{The meshed area describes the range of $h(\tau_1)$ and $h(\tau_2)$ such that the type II testing error may occur when 
$H(\san{t}_{\bm{x}_1+\bm{y}_1}) \simeq h(\tau_1)$ and $H(\san{t}_{\bm{x}_2+\bm{y}_2}) \simeq h(\tau_2)$ for (a) SHA scheme and (b) modified SHA scheme.}
\label{Fig:range-product1}
}
\end{figure}

\section{Discussion}

In this paper, we have developed tools to evaluate the critical rate attainable by the SHA scheme,
and exemplified that the SHA scheme is sub-optimal for a product of BSDSs.
A future problem is to generalize the idea of sequential scheme in Section \ref{section:sequential}  
and develop a new achievability scheme for the distributed hypothesis testing.
Another future problem is a potential utility of the scheme in \cite{WeiKoc:19};
as we mentioned in Section \ref{section:introduction}, it has a potential to improve upon the SHA scheme,
but strict improvement is not clear since the expression involves complicated optimization over multiple parameters.

Lastly, it should be pointed out that certain kinds of ``structured" random coding is known to improve upon the naive random coding
in the context of the mismatched decoding problem; eg.~see \cite{Lapidoth96, ScaBarMarFab15} and references therein. 
Even though there is no direct connection between the two problems, 
it might be interesting to pursue a mathematical connection between these problems. 

\section*{Acknowledgment}

The author appreciates Marat Burnashev for valuable discussions at an early stage of the work.
The author also thanks Te Sun Han and Yasutada Oohama for fruitful discussions. 

\newpage

\bibliographystyle{../../09-04-17-bibtex/IEEEtran}
\bibliography{../../09-04-17-bibtex/reference.bib}

\begin{thebibliography}{10}
\providecommand{\url}[1]{#1}
\csname url@samestyle\endcsname
\providecommand{\newblock}{\relax}
\providecommand{\bibinfo}[2]{#2}
\providecommand{\BIBentrySTDinterwordspacing}{\spaceskip=0pt\relax}
\providecommand{\BIBentryALTinterwordstretchfactor}{4}
\providecommand{\BIBentryALTinterwordspacing}{\spaceskip=\fontdimen2\font plus
\BIBentryALTinterwordstretchfactor\fontdimen3\font minus
  \fontdimen4\font\relax}
\providecommand{\BIBforeignlanguage}[2]{{%
\expandafter\ifx\csname l@#1\endcsname\relax
\typeout{** WARNING: IEEEtran.bst: No hyphenation pattern has been}%
\typeout{** loaded for the language `#1'. Using the pattern for}%
\typeout{** the default language instead.}%
\else
\language=\csname l@#1\endcsname
\fi
#2}}
\providecommand{\BIBdecl}{\relax}
\BIBdecl

\bibitem{berger:79b}
T.~Berger, ``Decentralized estimation and decision theory,'' in \emph{Presented
  at the IEEE 7th Spring Workshop on Information Theory}, Mt. Kisco, NY,
  September 1979.

\bibitem{AhlCsi86}
R.~Ahlswede and I.~Csisz\'ar, ``Hypothesis testing with communication
  constraints,'' \emph{IEEE Trans. Inform. Theory}, vol.~32, no.~4, pp.
  533--542, July 1986.

\bibitem{han:87}
T.~S. Han, ``Hypothesis testing with multiterminal data compression,''
  \emph{IEEE Trans. Inform. Theory}, vol.~33, no.~6, pp. 759--772, November
  1987.

\bibitem{ZhaBer88}
Z.~Zhang and T.~Berger, ``Estimation via compressed information,'' \emph{IEEE
  Trans. Inform. Theory}, vol.~34, no.~2, pp. 198--211, March 1988.

\bibitem{amari-han:89}
S.~Amari and T.~S. Han, ``Statistical inference under multiterminal rate
  restrictions: A differential geometric approach,'' \emph{IEEE Trans. Inform.
  Theory}, vol.~35, no.~2, pp. 217--227, March 1989.

\bibitem{amari:89}
S.~Amari, ``Fisher information under restriction of {S}hannon information in
  multi-terminal situations,'' \emph{Annals of the Institute of Statistical
  Mathematics}, vol.~41, no.~4, pp. 623--648, 1989.

\bibitem{han-kobayashi:89}
T.~S. Han and K.~Kobayashi, ``Exponential-type error probabilities for
  multiterminal hypothesis testing,'' \emph{IEEE Trans. Inform. Theory},
  vol.~35, no.~1, pp. 2--14, January 1989.

\bibitem{ShaPap92}
H.~M.~H. Shalaby and A.~Papamarcou, ``Multiterminal detection with zero-rate
  data compression,'' \emph{IEEE Trans. Inform. Theory}, vol.~38, no.~2, pp.
  254--267, March 1992.

\bibitem{han-amari:98}
T.~S. Han and S.~Amari, ``Statistical inference under multiterminal data
  compression,'' \emph{IEEE Trans. Inform. Theory}, vol.~44, no.~6, pp.
  2300--2324, October 1998.

\bibitem{Amari:11}
S.~Amari, ``On optimal data compression in multiterminal statistical
  inference,'' \emph{IEEE Trans. Inform. Theory}, vol.~57, no.~9, pp.
  5577--5587, September 2011.

\bibitem{ueta-kuzuoka:14}
M.~Ueta and S.~Kuzuoka, ``The error exponent of zero-rate multiterminal
  hypothesis testing for sources with common information,'' in \emph{Proc.
  ISITA 2014}, Melbourne, Australia, 2014, pp. 559--563.

\bibitem{polyanskiy:12}
Y.~Polyanskiy, ``Hypothesis testing via a comparator,'' in \emph{Proc. IEEE
  Int. Symp. Inf. Theory 2012}, Cambridge, MA, 2012, pp. 2206--2210.

\bibitem{RahWag:12}
M.~S. Rahman and A.~B. Wagner, ``On the optimality of binning for distributed
  hypothesis testing,'' \emph{IEEE Trans. Inform. Theory}, vol.~58, no.~10, pp.
  6282--6303, October 2012.

\bibitem{XiaKim:13}
Y.~Xiang and Y.-H. Kim, ``Interactive hypothesis testing against
  independence,'' in \emph{IEEE International Symposium on Information Theory},
  2013, pp. 2840--2844.

\bibitem{ZhaLai:18}
W.~Zhao and L.~Lai, ``Distributed testing with cascaded encoders,'' \emph{IEEE
  Trans. Inform. Theory}, vol.~64, no.~11, pp. 7339--7348, November 2018.

\bibitem{Wat:18}
S.~Watanabe, ``{N}eyman-{P}earson test for zero-rate multiterminal hypothesis
  testing,'' \emph{IEEE Trans. Inform. Theory}, vol.~64, no.~7, pp. 4923--4939,
  July 2018.

\bibitem{WeiKoc:19}
N.~Weinberger and Y.~Kochman, ``On the reliability function of distributed
  hypothesis testing under optimal detection,'' \emph{IEEE Trans. Inform.
  Theory}, vol.~65, no.~8, pp. 4940--4965, August 2019.

\bibitem{SalWigWan:19}
S.~Salehkalaibar, M.~Wigger, and L.~Wang, ``Hypothesis testing over the two-hop
  relay network,'' \emph{IEEE Trans. Inform. Theory}, vol.~65, no.~7, pp.
  4411--4433, July 2019.

\bibitem{HadLiuPolSha:19}
U.~Hadar, J.~Liu, Y.~Polyanskiy, and O.~Shayevitz, ``Communication complexity
  of estimating correlations,'' in \emph{Proceedings of the 51st Annual ACM
  Symposium on Theory of Computing (STOC '19)}.\hskip 1em plus 0.5em minus
  0.4em\relax ACM Press, 2019, pp. 792--803.

\bibitem{HadSha:19}
U.~Hadar and O.~Shayevitz, ``Distributed estimation of {G}aussian
  correlations,'' \emph{IEEE Trans. Inform. Theory}, vol.~65, no.~9, pp.
  5323--5338, September 2019.

\bibitem{EscWigZai:20}
P.~Escamilla, M.~Wigger, and A.~Zaidi, ``Distributed hypothesis testing:
  Cooperation and concurrent detection,'' \emph{IEEE Trans. Inform. Theory},
  vol.~66, no.~12, pp. 7550--7564, December 2020.

\bibitem{Bur:20}
M.~V. Burnashev, ``New upper bounds in the hypothesis testing problem with
  information constraints,'' \emph{Problems of Information Transmission},
  vol.~56, no.~2, pp. 64--81, 2020.

\bibitem{SreGun:20}
S.~Sreekumar and D.~G\"und\"uz, ``Distributed hypothesis testing over discrete
  memoryless channels,'' \emph{IEEE Trans. Inform. Theory}, vol.~66, no.~4, pp.
  2044--2066, April 2020.

\bibitem{SahTya:21}
K.~R. Sahasranand and H.~Tyagi, ``Communication complexity of distributed high
  dimensional correlation testing,'' \emph{IEEE Trans. Inform. Theory},
  vol.~67, no.~9, pp. 6082--6095, September 2021.

\bibitem{ShiHanAma94}
H.~Shimokawa, T.~S. Han, and S.~Amari, ``Error bound of hypothesis testing with
  data compression,'' in \emph{Proceedings of 1994 IEEE International Symposium
  on Information Theory}, 1994, p. 114.

\bibitem{HaiKoch:16}
E.~Haim and Y.~Kochman, ``Binary distributed hypothesis testing via
  {K}\"orner-{M}arton coding,'' in \emph{2016 IEEE Information Theory
  Workshop}, 2016.

\bibitem{csiszar-shields:04book}
I.~Csisz\'ar and P.~C. Shields, \emph{Information Theory and Statistics: A
  Tutorial}.\hskip 1em plus 0.5em minus 0.4em\relax Now Publishers Inc., 2004.

\bibitem{KuzWat16}
S.~Kuzuoka and S.~Watanabe, ``On distributed computing for functions with
  certain structures,'' \emph{IEEE Trans. Inform. Theory}, vol.~63, no.~11, pp.
  7003--7017, November 2017.

\bibitem{Lapidoth96}
A.~Lapidoth, ``Mismatched decoding and the multiple-access channel,''
  \emph{IEEE Trans. Inform. Theory}, vol.~42, no.~5, pp. 1439--1452, September
  1996.

\bibitem{ScaBarMarFab15}
J.~Scarlett, A.~Somekh-Baruch, A.~Martinez, and A.~G. i~F\'abregas, ``A
  counter-example to the mismatched decoding converse for binary-input discrete
  memoryless channels,'' \emph{IEEE Trans. Inform. Theory}, vol.~10, no.~10,
  pp. 5387--5359, October 2015.

\end{thebibliography}
\end{document}